\documentclass[10pt,a4paper]{article}

\usepackage[utf8]{inputenc}
\usepackage[english]{babel}

\usepackage{bussproofs}
\usepackage{amsthm} 
\usepackage{amsfonts}
\usepackage {tikz}
\usepackage{mathrsfs}
\usepackage{booktabs}
\usepackage{multirow}

\usepackage{amssymb}
\usepackage{graphicx}

\usetikzlibrary{fit,shapes.geometric, arrows}
\usetikzlibrary{positioning}
\usetikzlibrary{calc,shapes.multipart,chains}
\definecolor {processblue}{cmyk}{0.96,0,0,0}

\usepackage{subcaption}

\theoremstyle{definition} 
\newtheorem{definition}{Definition}[section] 
\theoremstyle{plain} 

\newtheorem{observation}{Observation}[section]

\usepackage[scr=boondox,  
            cal=esstix]   
           {mathalpha}

\usepackage{appendix} 

\usepackage{quoting} 
\quotingsetup{font=small}

\usepackage[square,sort,comma]{natbib}
\usepackage{varioref} 

\usepackage{float}

\usepackage{cancel}

\usepackage{hyperref} 
\usepackage{geometry}
\usepackage{bussproofs}
\usepackage{amsthm} 
\usepackage{amsfonts}
\usepackage {tikz}

\usetikzlibrary{fit,shapes.geometric, arrows}
\usetikzlibrary{positioning}
\usetikzlibrary{calc,shapes.multipart,chains}
\usepackage {xcolor}
\definecolor {processblue}{cmyk}{0.96,0,0,0}

\usepackage{comment}

\usepackage{subcaption}

\newcommand*{\nc}[2]{#1\mathbin{\left| \sim \vphantom{#1#2} \right.}#2}%

\makeatletter

\newcommand{\tright}{\mathrel\triangleright}
\DeclareRobustCommand{\btleft}{\mathrel{\mathpalette\btlr@\blacktriangleleft}}
\DeclareRobustCommand{\btright}{\mathrel{\mathpalette\btlr@\blacktriangleright}}

\newcommand{\btlr@}[2]{%
  \begingroup
  \sbox\z@{$\m@th#1\triangleright$}%
  \sbox\tw@{\resizebox{1.1\wd\z@}{1.1\ht\z@}{\raisebox{\depth}{$\m@th#1\mkern-1mu#2$}}}%
  \ht\tw@=\ht\z@ \dp\tw@=\dp\z@ \wd\tw@=\wd\z@
  \copy\tw@
  \endgroup
}

\makeatletter
\newcommand\bigDiamond{\mathop{\mathpalette\bigDi@mond\relax}}
\newcommand\bigDi@mond[2]{%
  \vcenter{\hbox{\m@th
    \scalebox{\ifx#1\displaystyle 2\else1.2\fi}{$#1\Diamond$}%
  }}%
}
\makeatother

\usepackage{listings}
\title{A Proof System with Causal Labels (Part I): checking Individual Fairness and Intersectionality}

\author{Leonardo Ceragioli and Giuseppe Primiero\footnote{LUCI Lab, Department of Philosophy, Università degli Studi di Milano}} 

\date{}

\begin{document}

\maketitle

\begin{abstract}
In this article we propose an extension to the typed natural deduction calculus \textbf{TNDPQ} to model verification of counterfactual fairness in probabilistic classifiers. This is obtained formulating specific structural conditions for causal labels and checking that evaluation is robust under their variation.
\end{abstract}


\section{Introduction}

The calculus \textbf{TPTND} (\textit{Trustworthy Probabilistic Typed Natural Deduction} \cite{10.1093/logcom/exaf003,KUBYSHKINA2024109212}) is designed to evaluate \textit{post-hoc} the trustworthiness of the behavior of opaque systems. The system is implemented for verification of dataframes in the tool BRIO \cite{DBLP:conf/beware/CoragliaDGGPPQ23,coraglia2024evaluatingaifairnesscredit}. In \cite{CeragioliPrimiero2025}, we introduced \textbf{TNDPQ} (\textit{Typed Natural Deduction for Probabilistic Queries}), a variation of the previous system in which a probabilistic output is associated to a target variable when a Data Point -- consisting of a list of values attributions for a set of variables -- is provided.
In this paper, we extend this system with tools to verify individual fairness of classifiers.

We start with a formal definition of classifiers.
Let $\mathscr{A}$ be a set of protected variables $ a_{1} , \ldots , a_{n}$, $\mathscr{X}$ be a (disjoint) set of non-protected variables $ x_{1} , \ldots , x_{n}$ and $t$ be a target variable.
Moreover, let $\mathscr{V}_{a_{i}}$ be a set of values $\alpha ^{i_{1}}, \alpha ^{i_{2}}, \ldots , \alpha ^{i_{j}}$ that $a_{i}$ can receive, $\mathscr{V}_{A}$ the set of all $\mathscr{V}_{a_{i}}$, $\mathscr{V}_{x_{i}}$ be a set of values $\beta ^{i_{1}}, \beta ^{i_{2}}, \ldots , \beta ^{i_{j}}$ that $x_{i}$ can receive, $\mathscr{V}_{X}$ the set of all $\mathscr{V}_{x_{i}}$, and $\mathscr{V}_{t}$ the set $\delta ^{1}, \delta ^{2}, \ldots , \delta ^{j}$ of values that $t$ can receive.
Let us use $ v_{1} , \ldots , v_{n}$ to denote elements of $\mathscr{A} \cup \mathscr{X}$ (that is, variables regardless of their protected or unprotected status), and $\gamma ^{i_{1}}, \ldots , \gamma ^{i_{j}}$ to denote the values that $v_{i}$ can receive.
We use $a_{i}:\alpha ^{i_{j}}$ (respectively $x_{i}:\beta ^{i_{j}}$) to express the \textit{judgment} that variable $a_{i}$ receives value $\alpha ^{i_{j}}$ (respectively, variable $x_{i}$ receives value $\beta ^{i_{j}}$), and $t:\delta ^{1}_{p_{1}}, \ldots ,\delta ^{j}_{p_{j}} $ to express the \textit{probabilistic judgment} that $\delta ^{1}, \ldots , \delta ^{j}$ are all the possible values that variable $t$ can receive and that, for $1 \leq k \leq j$, it receives value $\delta ^{k}$ with probability $p_{k}$.\footnote{
We assume that the values for $t$ (also for the elements of $\mathscr{A}$ and $\mathscr{X}$ but this is irrelevant here) are all mutually exclusive, and so $\sum _{k=1}^{j} p_{k}=1$.
Note that we make no assumption regarding whether $t\in \mathscr{A} \cup \mathscr{X} $ and so on whether $\mathscr{V}_{t} = \mathscr{V}_{a_{i}} $ or $ \mathscr{V}_{t} = \mathscr{V}_{x_{i}} $ for some $i$.}
We use $\mathscr{J}^{\mathscr{A}}$ for the set of all the judgments about protected variables, $\mathscr{J}^{\mathscr{X}}$ for the set of all the judgments about non-protected variables, and $\mathscr{J}^{\mathscr{P}}$ for the set of all probabilistic judgments.
Moreover, we use $\sigma ^{\mathscr{A}}$ to express a set of judgments about protected variables such that each element of $\mathscr{A}$ receives at most one value, $\sigma ^{\mathscr{X}}$ to express a set of judgments about non-protected variables such that each element of $\mathscr{X}$ receives at most one value, $\Sigma ^{\mathscr{A}}$ to refer to the set of all $\sigma ^{\mathscr{A}}$, and $\Sigma ^{\mathscr{X}}$ to refer to the set of all $\sigma ^{\mathscr{X}}$.
$\sigma$ is used to express the union of a $\sigma ^{\mathscr{A}}$ and a $\sigma ^{\mathscr{X}}$, and $\Sigma$ is used to refer to the set of all $\sigma$.
More formally:
\[
\Sigma ^{\mathscr{A}} =_{def} \{\sigma ^{\mathscr{A}} \subseteq \mathscr{J}^{\mathscr{A}} \mid \forall i ( a_{i}: \alpha ^{i_{l}} \in \sigma ^{\mathscr{A}} \land a_{i}: \alpha ^{i_{m}} \in \sigma ^{\mathscr{A}} \rightarrow l=m) \}
\]
\[
\Sigma ^{\mathscr{X}} =_{def} \{\sigma ^{\mathscr{X}} \subseteq \mathscr{J}^{\mathscr{X}} \mid \forall i ( x_{i}: \beta ^{i_{l}} \in \sigma ^{\mathscr{X}} \land x_{i}: \beta ^{i_{m}} \in \sigma ^{\mathscr{X}} \rightarrow l=m) \}
\]
\[
\Sigma =_{def} \{\sigma ^{\mathscr{A}} \cup \sigma ^{\mathscr{X}} \mid \sigma ^{\mathscr{A}} \in \Sigma ^{\mathscr{A}} ~ \land ~ \sigma ^{\mathscr{X}} \in \Sigma ^{\mathscr{X}} \}
\]

\noindent
A classifier $\widehat{\mathscr{f}} \in \widehat{\mathscr{F}} $ is a function from $\Sigma$ to $\mathscr{J}^{\mathscr{P}}$, where each $\sigma \in \Sigma$ describes a Data Point, that is what we know about a subject, and the probabilistic judgment $t: \delta ^{1} _{p_{1}}, \ldots ,\delta ^{j} _{p_{j}} $ in $\mathscr{J}^{\mathscr{P}}$ represents the output of the classifier regarding the probability distribution of the possible values for the target variable $t$.

\textbf{TNDPQ} is a proof system working with sequents describing the result of queries for classifiers.
More precisely, each classifier $\widehat{\mathscr{f}}$ is characterized by a set of ground sequents of the form:\footnote{
Technically, we should add a subscript in the equation specifying the classifier we are focusing on.
However, this will not be needed here, since we will not compare outputs of different classifiers.
}

\begin{equation}
\label{eq:judgments}
    \nc{\sigma}{t: \delta _{p}}
\end{equation}

\noindent
For readability reasons, the sequents focus on only one possible value for the target variable at a time.
In \cite{CeragioliPrimiero2025} we show how to extend \textbf{TNDPQ} with sequents working with logically complex judgments -- possibly non-atomic variables receiving possibly non-atomic values.
As an example, the following sequent expresses the probability that a non-white 27 years old woman who is married or divorced receives a loan:
\[
\nc{Age: 27 ,\: Gen.: f,\: MS: married + divorced,\: Etn.: white^{\bot}}{Loan : yes _{0.60}}
\]
\indent \textbf{TNDPQ} was initially designed to investigate the preservation of trustworthiness under the composition of logically simpler queries.
In this paper, we focus only on the atomic fragment and provide a \textit{criterion} to verify individual fairness for a probabilistic classifier via structural properties. Moreover, we address the issue of intersectionality, showing a solution through an extension of the original language with causal relations.

\section{Individual Fairness}

\begin{table*}
\centering
  \caption{
  The table shows the $680$ Data Points in the Training Set that satisfy $\sigma$.
  Of them: 
  $100$ satisfy $a_{1}: \alpha ^{1_{1}} $ and $a_{2}: \alpha ^{2_{1}} $, $90$ of which satisfy also $t:\delta $~; 
  $240$ satisfy $a_{1}: \alpha ^{1_{1}} $ and $a_{2}: \alpha ^{2_{2}} $, $180$ of which satisfy also $t:\delta $~;
  $240$ satisfy $a_{1}: \alpha ^{1_{2}} $ and $a_{2}: \alpha ^{2_{1}} $, $180$ of which satisfy also $t:\delta $~;
  $100$ satisfy $a_{1}: \alpha ^{1_{2}} $ and $a_{2}: \alpha ^{2_{2}} $, $90$ of which satisfy also $t:\delta $~.  
  By summing up the points in each column and in each row we obtain that $79 \%$ of the points satisfy $t:\delta $, so the lack of fairness disappears when at most one of $a_{1} $ or $a_{2} $ is considered.
  }
  \label{tab:IFInter}
  \begin{tabular}{cc||c|c||c}
\toprule
 & & \multicolumn{2}{c}{attribute $a_{1}$} & \\
& & $\alpha ^{1_{1}}$ & $\alpha ^{1_{2}}$ & \\
\midrule
\midrule
\multirow{2}*{attribute $a_{2}$} & $ \alpha ^{2_{1}} $ & $\frac{90}{100}\approx0.90$ & $\frac{180}{240}\approx0.75$ & $\frac{270}{340}\approx0.79$ \\
\cmidrule(lr){2-5}
 							& $ \alpha ^{2_{2}} $ & $\frac{180}{240}\approx0.75$ & $\frac{90}{100}\approx 0.90$ & $\frac{270}{340}\approx0.79$ \\
\midrule
\midrule
&  & $\frac{270}{340}\approx 0.79$ & $\frac{270}{340}\approx 0.79$ &  \\
 \bottomrule
\end{tabular} 
\end{table*}

Individual fairness has different non-equivalent definitions in the literature.
A first simplified characterization of \textit{individual fairness} is as follows:

\begin{definition}[Individual Fairness (\textbf{IF})]
\label{def:IF}
    A classifier is individually fair regarding a set of protected attributes if it gives the same outputs to Data Points differing only for the values of those attributes.
    Formally, $\widehat{\mathscr{f}}$ is \textbf{IF} regarding the set of protected attributes $\{a_{1}, \ldots , a_{n}\}$ iff for every $\sigma ^{\mathscr{X}} \in \Sigma ^{\mathscr{X}} $ and pair of n-tuples of values $\alpha ^{1_{l}}, \ldots , \alpha ^{n_{l}} $ and $ \alpha ^{1_{m}}, \ldots , \alpha ^{n_{m}} $, 
    $\widehat{\mathscr{f}} (\sigma ^{\mathscr{X}} , a_{i}:\alpha ^{1_{l}}, \ldots , a_{i}:\alpha ^{n_{l}} ) = \widehat{\mathscr{f}} (\sigma ^{\mathscr{X}} , a_{i}:\alpha ^{1_{m}}, \ldots , a_{i}:\alpha ^{n_{m}} )$.
\end{definition}

\noindent
Given definition~\ref{def:IF}, in \textbf{TNDPQ} a classifier is \textbf{IF} regarding the set of protected attributes $\{a_{1}, \ldots , a_{n}\}$ iff the ground sequents describing the classifier are such that  
$\nc{\sigma ^{\mathscr{X}}, a_{i}:\alpha ^{1_{l}}, \ldots , a_{i}:\alpha ^{n_{l}}}{t: \delta _{p}}$ iff $\nc{\sigma ^{\mathscr{X}} , a_{i}:\alpha ^{1_{l}}, \ldots , a_{i}:\alpha ^{n_{l}}}{t: \delta _{p}}$, for every $\delta \in \mathscr{V}_{t}$, $\alpha ^{i_{l}},\alpha ^{i_{m}} \in \mathscr{V}_{a_{i}}$, and $\sigma ^{\mathscr{X}} \in \Sigma ^{\mathscr{X}}$.

We consider \textbf{IF} as the best way of approximating \textit{fairness through unawareness} when we have to evaluate opaque systems:

\begin{definition}[Fairness through Unawareness (\textbf{FtU})]
    A classifier is fair through unawareness regarding a protected attribute as long as this attribute is not explicitly used in the decision-making process.
\end{definition}

\noindent
Indeed, while \textbf{FtU} is clearly an \textit{intensional} notion, which can be properly evaluated only by looking at the implemented program, \textbf{IF} can be evaluated just by looking at the inputs and outputs of the machine.
For this reason, we cannot directly evaluate \textbf{FtU} for opaque systems, and \textbf{IF} emerges as a good substitute.

However, there is a problem.
Although intersectionality holds for \textbf{FtU}, it fails for \textbf{IF}.\footnote{
Notice that, even though we focus only on the case of two protected variables, the result generalizes for any set of protected variables.
}

\begin{observation}[Intersectionality fails for \textbf{IF}]
A classifier that is \textbf{IF} with respect to protected attribute $a_{1}$ and protected attribute $a_{2}$ (separately) may not be \textbf{IF} regarding the set $\{a_{1}, a_{2}\}$.

\end{observation}

\begin{proof}
    To prove failure of intersectionality, we just have to show that for some classifier $\widehat{\mathscr{f}}$, for every $\sigma ^{\mathscr{X}} \in \Sigma ^{\mathscr{X}} $, for every pair of values $\alpha ^{1_{l}} $ and $ \alpha ^{1_{m}}$ in $ \mathscr{V}_{a_{1}}$, and for every pair of values $\alpha ^{2_{l}} $ and $ \alpha ^{2_{m}}$ in $ \mathscr{V}_{a_{2}}$
    \[
    \widehat{\mathscr{f}} (\sigma ^{\mathscr{X}} , a_{1}:\alpha ^{1_{l}}) = \widehat{\mathscr{f}} (\sigma ^{\mathscr{X}} , a_{1}:\alpha ^{1_{m}})
    \]
    \[
    \widehat{\mathscr{f}} (\sigma ^{\mathscr{X}} , a_{2}:\alpha ^{2_{l}}) = \widehat{\mathscr{f}} (\sigma ^{\mathscr{X}} , a_{2}:\alpha ^{2_{m}})
    \]
    \noindent
    but for some pair of sets of values $\alpha ^{1_{l}} , \alpha ^{2_{l}}$ and $\alpha ^{1_{m}} , \alpha ^{2_{m}}$
    \[
    \widehat{\mathscr{f}} (\sigma ^{\mathscr{X}} , a_{1}:\alpha ^{1_{l}}, a_{2}:\alpha ^{2_{l}}) \neq \widehat{\mathscr{f}} (\sigma ^{\mathscr{X}} , a_{1}:\alpha ^{1_{m}}, a_{2}:\alpha ^{2_{m}})
    \]
    We will show that such a classifier is not only theoretically possible, but even quite common when ML systems are trained using Data Sets of a specific kind.

    For simplicity, assume that $a_{1} $ and $a_{2} $ have only two possible outputs each (respectively $\alpha ^{1_{1}}$ and $\alpha ^{1_{2}}$, and $\alpha ^{2_{1}}$ and $\alpha ^{2_{2}}$).
    Let us consider an ML system implementing a learning algorithm with no restriction on protected attributes.
    The system is not \textbf{FtU} regarding these attributes, but can be \textbf{IF} if it is trained using a fair Data Set: that is, a Data Set in which all the Data Points sharing the same value of the non-protected attributes but possibly differing for those of $a_{1} $ or $a_{2} $ share the same value of the target variable.
    In our case, let us assume that the Data Set is fair in this sense and focus on the specific Data Points in table~\ref{tab:IFInter}: these are all the Data Points that satisfy a specific $\sigma ^{\mathscr{X}} \in \Sigma ^{\mathscr{X}}$, with the ratio expressing how many of them give value $\delta$ to target variable $t$.
    We can observe, by looking at the table, that the Data Points are fair when $a_{1}$ and $a_{2}$ are considered separately, and biased when $a_{1}$ and $a_{2}$ are considered together.
    Hence, the ML system will not learn to be \textbf{IF} regarding $\{a_{1},a_{2}\}$.
    As an example: 
    \[\nc{\sigma ^{\mathscr{X}} , a_{1}: \alpha ^{1_{1}}, a_{2}: \alpha ^{2_{1}}}{t: \delta _{0.90}}\]
    \[\nc{\sigma ^{\mathscr{X}} , a_{1}: \alpha ^{1_{1}}, a_{2}: \alpha ^{2_{2}}}{t: \delta _{0.75}}\]
    Note that a different probability associated with just one value of $t$ is sufficient to disprove \textbf{IF} for the set of variables $\{a_{1},a_{2}\}$ and so the proof is complete.
\end{proof}

Lack of intersectionality is even more serious than it could seem.
Indeed, if intersectionality fails, fairness can be gerrymandered by cherry picking protected attributes, so this property contributes to make \textbf{IF} relevant even if the focus is only on single protected attributes \cite{pmlr-v80-kearns18a,Intersect2025}.
Moreover, intersectionality also fails for more elaborated notions of individual fairness which implement a metric for similarity of Data Points and require similar predictions for similar points \cite{kusner2017a,Asher2022}.
In fact, while a shared assumption in the existing literature about \textbf{IF} is that all features of the Data Points are mutually independent, in the next section we argue that the causal relations among such features must be taken into account in order to check intersectionality for opaque systems. 

\section{IF and intersectionality as Weakening}

Since \textbf{IF} and intersectionality require that the probability of a sequent does not change when different values are attributed to protected variables, both these properties can be seen as equivalent to a restricted rule of \textit{Weakening}:\footnote{Notice that intersectionality is addressed by considering $\sigma$ and not only $\sigma ^{\mathscr{X}}$.}

\begin{equation}
\label{rule:Weak1}
    \AxiomC{$ \nc{\sigma  }{t:\delta _{p}}$}
	\RightLabel{{\tiny $Weakening ^{*}$}}
	\UnaryInfC{$ \nc{\sigma , a: \alpha }{t:\delta _{p}} $}
    \DisplayProof
\end{equation}

\noindent
Hence, since \textbf{TNDPQ} is non-monotonic, we need to provide the \textit{condition} under which this rule is valid, in order to deal with \textbf{IF} and intersectionality.
As an example, the following inference establishes \textbf{IF} regarding the protected attribute \textit{gender} (variable $Gen.$) and an instance of intersectionality in case also \textit{marital status} ($MS$) is considered protected:

\begin{prooftree}
    \AxiomC{$ \nc{Age: 27 ,\: MS: married + divorced,\: Etn.: white^{\bot}}{Loan : yes _{0.60}}$}
	\RightLabel{{\tiny $Weakening ^{*}$}}
	\UnaryInfC{$ \nc{Age: 27 ,\: Gen.: f,\: MS: married + divorced,\: Etn.: white^{\bot}}{Loan : yes _{0.60}}$}
\end{prooftree}

Since \textbf{TNDPQ} is a probabilistic system, \textit{Weakening} is valid when the active variable in the rule ($a$) and the target variable ($t$) are mutually independent, conditional on $\sigma$. 

\begin{definition}[Conditional Independence]
$t$ and $a$ are independent, conditional on $\sigma$, iff 
\[
P(t: \delta \mid a: \alpha , \sigma) = P(t: \delta \mid \sigma)
\]
\end{definition}

\noindent
Therefore, what we need is a \textit{criterion} of conditional independence.

In a logic that ignores the relations between the elements of $\sigma$, conditional independence can be decided only using brute-force methods, i.e. by checking statistical correlations for all values of any variable.
Moreover, we have no way of distinguishing good correlations from spurious ones which may emerge due to biases in the Training Set.
In contrast, a system extended with the basic vocabulary of causal graphs \cite{pearl2017} can describe directly causal relations between the features of the classifier.
Hence, we translate in our calculus some well-studied conditions for independence, obtaining admissibility \textit{criteria} for Weakening, in turn establishing \textbf{IF} and intersectionality.
For the purposes of this work, we define causal graphs as follows:

\begin{definition}[Causal Graph]
    A causal graph is an acyclic directed graph with nodes representing events (variables receiving values) and edges representing immediate causal relations. 
\end{definition}

\noindent
By closing edges under transitivity, we obtain the notion of mediate cause.
For purely formal reasons, we close the notion of cause under reflexivity as well.
The usual extension of deterministic causal graphs with functions to compute the value of a node on the basis of those of all the immediate parent nodes is here expressed by sequents like in equation~\ref{eq:judgments}.
Moreover, the common distinction between exogenous and endogenous nodes, relevant in discussing interventions and counterfactual fairness, is left for further research. 

Two nodes which are one the immediate cause of the other are mutually dependent. While for two nodes which are not directly connected, dependence is defined in three steps \cite{pearl2017}:
\begin{itemize}
    \item The \textit{criteria} in figure~\ref{fig:causal} deal with the easiest cases possible, when there is only one intermediate node between the two;
    \item We define a path as blocked by a set of nodes iff, when all and only these nodes occur in the condition, at least one chain, fork or collider in the path has independent nodes;
    \item Two nodes are independent iff all the paths between them are blocked.
\end{itemize}

\begin{figure}
\centering
\begin{subfigure}[b]{0.3\textwidth}
\centering
\captionsetup{width=0.95\textwidth}
\begin{tikzpicture}[-latex ,auto ,node distance =1 cm and 1.5cm ,on grid ,
semithick ,
state/.style ={ circle ,top color =white , bottom color = black!20 ,
draw, text=black , minimum width =0.5 cm}]

\node[state] (X) {i};
\node[state] (Y) [right =of X] {j};
\node[state] (Z) [right=of Y] {k};

\path (X) edge (Y);
\path (Y) edge (Z);

\end{tikzpicture}
\caption{\textbf{Chain:} $i$ and $k$ are dependent, but independent conditional on $j$.}
\end{subfigure}
\begin{subfigure}[b]{0.3\textwidth}
\centering
\captionsetup{width=0.95\textwidth}
\begin{tikzpicture}[-latex ,auto ,node distance =1 cm and 1.5cm ,on grid ,
semithick ,
state/.style ={ circle ,top color =white , bottom color = black!20 ,
draw , text=black , minimum width =0.5 cm}]

\node[state] (Y) {i};
\node[state] (X) [above right =of Y] {j};
\node[state] (Z) [below right =of X] {k};

\path (X) edge (Y);
\path (X) edge (Z);

\end{tikzpicture}
\caption{\textbf{Fork:} $i$ and $k$ are dependent, but independent conditional on $j$.}
\end{subfigure}
\begin{subfigure}[b]{0.3\textwidth}
\centering
\captionsetup{width=0.95\textwidth}
\begin{tikzpicture}[-latex ,auto ,node distance =1 cm and 1.5cm ,on grid ,
semithick ,
state/.style ={ circle ,top color =white , bottom color = black!20 ,
draw , text=black , minimum width =0.5 cm}]

\node[state] (X) {i};
\node[state] (Z) [below right =of X] {j};
\node[state] (Y) [above right=of Z] {k};

\path (X) edge (Z);
\path (Y) edge (Z);

\end{tikzpicture}
\caption{\textbf{Collider:} $i$ and $k$ are dependent conditional on $j$ or any of its descendant, independent otherwise.}
\end{subfigure}
\caption{Elementary compositions of nodes in causal graphs and \textit{criteria} of conditional independence.}
\label{fig:causal}
\end{figure}

To express the \textit{criteria} for conditional independence and thus for the applicability of the rule of Weakening, we need to internalize the causal notions in our calculus \textbf{TNDPQ}.
For this purpose, we use the methodology of labeled calculi \cite{Vigano2000,NegriVonPlato2011}.
First, we extend the language with the following relational predicates for variables:

\begin{description}
\item[Immediate Causal Relations] $ v_{i} \tright v_{j} =_{def}$ $v_{i}$ is an immediate cause of $v_{j}$.

\item[Mediate Causal Relations] $ v_{i} \btright ^{M} v_{j} =_{def}$ $v_{i}$ is a mediate cause of $v_{j}$, with intermediate nodes $M$.

\item[Path with Intermediate Nodes] $ v_{i} \bigDiamond ^{M}_{N} v_{j} =_{def}$ a path exists between $v_{i}$ and $v_{j}$ passing through non-colliders $M$ and colliders or sets of their descendants $N$.
\end{description}

\noindent
then, we reformulate \textbf{TNDPQ} sequents by extending their left-hand side with causal relations:

\begin{equation}
    \nc{\tright _{\widehat{\mathscr{f}}}~, \sigma}{t: \delta _{p}}
\end{equation}

\noindent
Let us use $Var_{\sigma}$ to indicate the set of variables that occur in $\sigma$. We use $\tright _{\widehat{\mathscr{f}}}~$ to indicate all the immediate causal relations among features in the classifier. $\bigDiamond _{\widehat{\mathscr{f}}}~$ denotes all the existing paths in the resulting graph and is derivable as the closure of $\tright _{\widehat{\mathscr{f}}}~$ under the rules in table~\ref{tab:Correl}.

To see how these rules work, consider the figure~\ref{fig:path}.
The rules for mediate causal relations are used to identify the descendants of colliders.
In this case, applying \textbf{reflexive cause} and \textbf{transitive cause} we obtain: $1 \btright ^{\{1,1b,1c\}} 1c$, $3 \btright ^{\{3,3b,3c\}} 3c$, and $3 \btright ^{\{3,3b,3d\}} 3d$.
The rules \textbf{chain}, \textbf{fork}, and \textbf{collider} are used to represent triplets of nodes between $a$ and $t$, with the rule for collider representing also the descendants.
Chains and forks are stored in the superscript: $1 \bigDiamond ^{\{2\}} 3$, $3 \bigDiamond ^{\{4\}} t$.
The colliders are stored in the subscript, together with the set of their descendants: $a \bigDiamond _{\{1,\{1,1b,1c\}\}} 2$, $2 \bigDiamond _{\{3,\{3,3b,3c\}\}} 4$, and $2 \bigDiamond _{\{3,\{3,3b,3d\}\}} 4$.\footnote{
Notice that the collider occurs both by itself and in the set of its descendants: its occurrence as collider is needed for the rule of transitivity, and its occurrence in the set of its descendants is needed for the condition of applicability of the rule~(\ref{rule:Weak2}).
}
\textbf{Transitivity} is used to combine these triplets to construct the paths between $a$ and $t$.
In this case, we have two paths: $a \bigDiamond ^{\{2,4\}}_{\{1,\{1,1b,1c\},3,\{3,3b,3c\}\}} t $, and $ a \bigDiamond ^{\{2,4\}}_{\{1,\{1,1b,1c\},3,\{3,3b,3d\}\}} t $.\footnote{
We focus only on the maximal sets of descendants, although technically also the paths containing only some of the descendants are constructible.
Notice that this does not cause problems with the conditions of rule~(\ref{rule:Weak2}).}

\begin{figure}
\centering
\begin{tikzpicture}[-latex ,auto ,node distance =1 cm and 1.5cm ,on grid ,
semithick ,
state/.style ={ circle ,top color =white , bottom color = black!20 ,
draw , text=black , minimum width =0.5 cm}]

\node[state] (a) {a};
\node[state] (1) [right =of a] {1};
\node[state] (2) [right=of 1] {2};
\node[state] (3) [right=of 2] {3};
\node[state] (4) [right =of 3] {4};
\node[state] (t) [right=of 4] {t};

\node[state] (1b) [below=of 1] {1b};
\node[state] (1c) [below=of a] {1c};

\node[state] (3b) [below=of 3] {3b};
\node[state] (3c) [below =of 2] {3c};
\node[state] (3d) [below =of 4] {3d};

\path (a) edge (1);
\path (2) edge (1);
\path (2) edge (3);
\path (4) edge (3);
\path (t) edge (4);

\path (1) edge (1b);
\path (1b) edge (1c);

\path (3) edge (3b);
\path (3b) edge (3c);
\path (3b) edge (3d);

\end{tikzpicture}
\caption{Paths between nodes $a$ and $t$.}
\label{fig:path}
\end{figure}

Note that in this calculus variables play the same role as labels in labeled calculi, with $\tright _{\widehat{\mathscr{f}}}~$ making explicit accessibility relations.
As an example, the sequent expressing that the probability that a 27 years old person with a gross annual income of $40.000 €$ receives a loan is 60\%, is formulated as follows:
\[
\nc{Age \tright MS ,Age \tright GAI , Age \tright Loan, GAI \tright Loan, Age: 27 ,\: GAI: 40K €}{Loan : yes _{0.60}}
\]

\begin{table*}[h!]
  \caption{Rules to derive paths $\bigDiamond _{\widehat{\mathscr{f}}}~$ from immediate causal relations $\tright _{\widehat{\mathscr{f}}}$.
  The rule of Transitivity has the condition that $i\in O \cup P$ and $j\in M \cup N$.}
  \label{tab:Correl}
\begin{tabular}{l l} 
\textbf{Reflexive cause} $\vdash v_{i} \btright ^{\{i\}} ~ v_{i}$ 
&
\textbf{Transitive cause} $v_{i} \btright ^{N} ~ v_{j},  v_{j} \tright v_{k} \vdash v_{i} \btright ^{N \cup \{ k\} } ~ v_{k} $ 
\\[0.2cm] 
\textbf{Chain} $v_{i} \tright v_{j}, v_{j} \tright v_{k} \vdash v_{i} \bigDiamond ^{\{j\}} v_{k} $
&
\textbf{Fork} $v_{j} \tright v_{i}, v_{j} \tright v_{k} \vdash v_{i} \bigDiamond ^{\{j\}} v_{k} $
\\[0.2cm] 
\textbf{Collider} $ v_{i} \tright v_{j}, v_{k} \tright v_{j}, v_{j} \btright ^{N} ~ v_{z} \vdash v_{i} \bigDiamond _{\{j,N\}} v_{k}$ 
&
\textbf{Transitivity$^{*}$} $v_{x} \bigDiamond ^{M}_{N} v_{i}, v_{j} \bigDiamond ^{O}_{P} v_{y} \vdash v_{x} \bigDiamond ^{M\cup O}_{N \cup P} v_{y}$
\end{tabular} 
\end{table*}

\indent 
With these technical tools in place, we formulate an applicability \textit{criterion} for the \textit{Weakening} rule in equation~\ref{rule:Weak2}: 

\begin{equation}
\label{rule:Weak2}
    \AxiomC{$\nc{\tright_{\widehat{\mathscr{f}}}~,\ \sigma}{t:\delta _{p}}$}
	\RightLabel{{\tiny $Weakening ^{*}$}}
	\UnaryInfC{$\nc{\tright _{\widehat{\mathscr{f}}}~,\ \sigma,\  a:\alpha}{t:\delta _{p}}$}
\DisplayProof
\end{equation}

\noindent
With \textit{Conditions} consisting of:
\begin{description}
\item[Condition1:] $a \not\triangleright ~t$ and $t \not\triangleright ~a$;
\item[Condition2:] For every $a \bigDiamond ^{M}_{N} t$ in $\bigDiamond_{\widehat{\mathscr{f}}}~$ , either $M \cap Var_{\sigma} \neq \emptyset$ or $\exists S \in N(\exists x \in S \land S \cap Var_{\sigma} = \varnothing)$.
\end{description}

\noindent
Hence, to decide whether Weakening can be applied, we check both $\tright _{\widehat{\mathscr{f}}}~$ and $\bigDiamond_{\widehat{\mathscr{f}}}~$ to evaluate conditional independence.
In particular, the first condition requires that protected variable and target variable are not one the direct cause of the other, and the second condition requires that, for every path connecting them, $\sigma$ blocks it, by containing at least one non-collider or by not containing a collider and all of its descendants.
Note that this rule can be used to decide both \textbf{IF} in general and intersectionality, which corresponds to cases in which $\sigma$ already contains a protected attribute.
More precisely, what is obtained by checking the admissibility of an instance of Weakening is an evaluation of fairness (and possibly intersectionality) of the classifier, when a specific set of attributes are used to decide a target variable.
Figure~\ref{fig:loan} shows a simple example of application of this rule.

\begin{figure}[h!]
\begin{subfigure}[b]{0.33\textwidth}
\captionsetup{width=0.95\textwidth}
\begin{tikzpicture}[-latex ,auto ,node distance =1.5 cm and 1.5cm ,on grid ,
semithick ,
state/.style ={ circle ,top color =white , bottom color = black!20 ,
draw , text=black , minimum width =0.5 cm}]

\node[state] (X) {Age};
\node[state] (Z) [below right =of X] {Loan};
\node[state] (Y) [above right=of Z] {GAI};
\node[state] (W) [below left=of X] {MS};

\path (X) edge (Z);
\path (Y) edge (Z);
\path (X) edge (Y);
\path (X) edge (W);

\end{tikzpicture}
\caption{Causal graph of the classifier.}
\end{subfigure}
\begin{subfigure}[b]{0.65\textwidth}
\captionsetup{width=0.95\textwidth}
\begin{flushright}
\AxiomC{$\nc{\tright _{\widehat{\mathscr{f}}}~, Age: 27 ,\: GAI: 40K €}{Loan : yes _{0.60}}$}
	\RightLabel{{\tiny $W ^{*}$}}
	\UnaryInfC{$\nc{\tright _{\widehat{\mathscr{f}}}~, Age: 27 ,\: GAI: 40K €,\: MS: m}{Loan : yes _{0.60}}$}
    \DisplayProof
\end{flushright}
\caption{Application of $W$ with attribute $MS$ receving value $m$ (married).
The satisfaction of the conditions can be observed from the causal graph, and is derivable from the set $\tright _{\widehat{\mathscr{f}}}$.}
\end{subfigure}
\caption{Example of application of Weakening.}
\label{fig:loan}
\end{figure}

\section{Conclusion}

This work focuses on formal tools to check fairness of probabilistic classifiers.
We have shown that, without taking into account the causal relations between the features of the classifier, intersectional fairness is not guaranteed.
The proposed typed natural deduction calculus \textbf{TNDPQ} has labels representing causal relations, and it provides a criterion of applicability for the rule of Weakening that establishes both fairness and intersectionality.
An extension using causal labels to express counterfactual fairness is left for further work.

\section*{Acknowledgments}

This research was supported by the Ministero dell’Università e della Ricerca (MUR) through PRIN 2022 Project SMARTEST – Simulation of Probabilistic Systems for the Age of the Digital Twin (20223E8Y4X), and through the Project “Departments of Excellence 2023-2027” awarded to the Department of Philosophy “Piero Martinetti” of the University of Milan.

\section*{Declaration on Generative AI}

 During the preparation of this work, the authors used Grammarly in order to: Grammar and spelling check. After using these tool, the authors reviewed and edited the content as needed and take full responsibility for the publication’s content.

\bibliography{biblio}


\end{document}